\documentclass[preprint,12pt]{elsarticle}

\usepackage{graphicx}
\usepackage[utf8]{inputenc}
\usepackage{lscape}
\usepackage{color} 
\usepackage{url}

\usepackage{latexsym} 
\usepackage{amsmath} 
\usepackage{amssymb}  
\usepackage{amsthm}

\usepackage{rotating}
\usepackage{color}           
\usepackage[all]{xy} 

\usepackage[english]{babel}

\usepackage{graphicx}
\usepackage[utf8]{inputenc}
\usepackage{lscape}
\usepackage{color} 
\usepackage{url}

\newcommand{\anillo}{{\mathbb F}_2[{\bf x}]}
\newcommand{\cuerpo}{{\mathbb F}_2}

\newcommand{\idos}{{\mathbb  I}_2}
\newcommand{\degree}{\mbox{deg}}
\newcommand{\liff}{\leftrightarrow}
\newcommand{\co}{\subseteq}

\newcommand{\Land}{\bigwedge}
\newcommand{\te}{\vdash}
\newcommand{\N}{{\mathbb N}}
\newcommand{\m}{\models}
\renewcommand{\mod}{\mbox{\rm mod \ }}
\newcommand{\lif}{\to}
\newcommand{\Liff}{\quad \Longleftrightarrow \quad}

\newcommand{\Lif}{\quad \Longrightarrow \quad}
\newcommand{\Lifn}{\quad \Longrightarrow \quad}

\newcommand{\fdos}{{\mathbb F}_2}

\theoremstyle{plain}

 \newcommand{\ttp}[1]{\color{black} #1}

 \definecolor{orange}{rgb}{1,0.5,0.1}

 \usepackage{latexsym} 
\usepackage{amsmath} 
\usepackage{amssymb}  
\usepackage{rotating}

\theoremstyle{plain}

\newtheorem{theorem}{Theorem}[section]
\newtheorem{lemma}[theorem]{Lemma}
\newtheorem{proposition}[theorem]{Proposition}
\newtheorem{corollary}[theorem]{Corollary}
\newtheorem{remark}[theorem]{Remark}

\newtheorem{definition}[theorem]{Definition}

\newtheorem{example}[theorem]{Example}

\usepackage{multicol}

\usepackage{pgfplots}
\pgfplotsset{compat=newest}

\begin{document}

\begin{frontmatter}

\title{A logic-algebraic tool for reasoning with Knowledge-Based Systems\footnote{This work was partially supported by  TIN2013-41086-P project (Spanish Ministry of Economy and Competitiveness), co-financed with FEDER funds.}}

\author{José A. Alonso-Jim\'{e}nez$^1$, Gonzalo A. Aranda-Corral$^2$, 
 Joaqu\'{\i}n Borrego-D\'{\i}az$^1$, M. Magdalena Fern\'{a}ndez-Lebr\'{o}n$^3$, M. Jos\'{e} Hidalgo-Doblado$^1$}

\address{$^1$Departamento de Ciencias de la Computaci\'{o}n e Inteligencia
  Artificial, E.T.S. Ingenier\'{\i}a Inform\'{a}tica,
    Universidad de Sevilla,  Avda. Reina
Mercedes s.n. 41012-Sevilla, Spain\\ 
  $^2$ Department of Information Technology, Universidad de Huelva
Crta. Palos de La Frontera s/n. 21819 Palos de La Frontera. Spain\\
$^3$Departamento de Matem\'{a}tica Aplicada I,
E.T.S. Ingenier\'{\i}a Inform\'{a}tica,
    Universidad de Sevilla,  Avda. Reina
Mercedes s.n. 41012-Sevilla, Spain}

\date{}

\begin{abstract}
A detailed  exposition of  foundations of  a logic-algebraic model
for reasoning with  knowledge bases specified by propositional  (Boolean)  logic is presented. The model is conceived from the logical translation of usual derivatives on polynomials (on residue rings) which is used to design   a new inference rule of algebro-geometric inspiration. Soundness and (refutational) completeness of the rule are proved. Some applications of the tools introduced in the paper are shown.
\end{abstract}

\begin{keyword}
Polynomial Semantics \sep  Symbolic Computing    \sep Automated Deduction \sep  Knowledge-Based Systems
\end{keyword}


\end{frontmatter}



\section{Introduction}

 Algebraic models for logic have been revealed as a useful tool for know\-led\-ge representation and mechanized reasoning. The relationship between certain algebraic structures and Computational Logic  provides methods and tools for building, compiling and reasoning with Knowledge-Based Systems (KBS) (see e.g.  \cite{racsam} for an introduction). This relationship also provides ma\-the\-ma\-tical foundations for a number of Knowledge Representation and Reasoning  (KRR) methods and algorithms, encompassing applications since  the pioneer works for classical bivalued logic \cite{kapur,hsiang} to  extensions for multi-valued logics \cite{jalonso,jalonso2,amai,polynomial}. The framework has also been extended to other logics for Artificial Intelligence (AI) as the paraconsistent logic \cite{janl}, and even towards other nonstandard reasoning tasks as the argument-based one (cf. \cite{mcs15}). One of the benefits of the interpretation of logic in polynomial rings  is that enables the use of powerful algebraic tools as 
 Gröbner Basis to compile KBS, exploiting this way the use of advanced Computer Algebra Systems in KRR.
 
 Roughly speaking, algebraic models for logic are mainly based on to specify, obtain and exploit  solutions for the logical entailment problem and other  related ones. Recall that the entailment problem in logic is stated as follows: given a Knowledge Base (KB) $K$, and  a formula $F$, to decide whether  $F$ is {a} logical consequence from $K$ (denoted by $K\models F$), that is, whether every model of $K$ is also model of $F$.

 This paper is focused in to expound with detail an(other) algebraic model for  KRR. 
Whilst aforementioned approaches do not need to design a new calculus (ideal membership translation -through Gröbner Basis- of entailment question is sufficient), the approach presented here consists of to design an inference rule -called {\em independence rule}- from an algebraic operation on polynomials\footnote{The part of the paper devoted to this is an extended version of \cite{calculemus}.}. This rule will allow address a number of other related problems.

 Although the independence rule  is inspired in Algebra, the idea is intimately related with KRR strategies which are oriented to  mitigate the complexity and size of the $KB$ (which may be of large size), with the hope of  reducing the computational cost of the deductive process when it is applied {\ttp  into specialized contexts} or use cases. Two strategies of this type are interesting for the purposes of the paper and motivate this work. 

The first is to facilitate the design of  divide-and-conquer strategies to deal with the entailment problem, a natural idea for managing KBs with hundreds of thousands of logical axioms with big size logical language (for example the strategy based on the reasoning with microtheories  \cite{CyC}). {\ttp In this case the problem of ensuring the completeness of the designed strategy arises, being a} critical issue the selection/synthesis of sound sub-KBs. In \cite{Amir} authors propose a partition-based strategy in order to obtain subtheories. It is based on a syntax level analysis that provides individual partitions where to reason locally, and distributed reasoning needs of methods to propagate information among different partitions. 

The second strategy to consider is based on the ad-hoc reduction {\ttp of KB for particular use cases} (distilling the KB for using in context-based reasoning, for example). In \cite{calculemus} authors propose to reduce $K$ to  $K'$, where $K\models K'$, and in $K'$ only the language of the goal formula $F$  is used (thus it is expected that the size of $K'$ to be  smaller than the size of $K$). Then the entailment problem with respect to $K'$ is considered. For this strategy to be successful -valid and complete- $K$ must be a {\em conservative extension} of $K'$ (or, equivalently, $K'$  a {\em conservative retraction} of $K$ \cite{calculemus}). A knowledge base $K'$ in the language $\mathcal L'$ is a conservative retraction of $K$ if $K$ is an extension of $K'$ such that every $\mathcal L'$-fórmula entailed by $K$ is also entailed by  $K'$.  Then the use of conservative retraction allows  to reduce the own KB we need to work, because it suffices to conservatively retracts the original KB to the specialized language of the formula-goal. A key question is how to compute such kind of sub-KBs.  

Whilst conservative extensions have been deeply investigated in several fields of Mathematical Logic and Computer Science (because they allow the formalization of several notions concerning refinements and modularity, see e.g. \cite{el+,mjoseh,fjesus}), solutions focused on its dual notion, the conservative retraction, are obstructed by its logical complexity (see e.g. \cite{lang}).

{\ttp Beside the analysis of above strategy, as a secondary motivation of the approach, it is worth to mention the study of  {\em relevance}  in knowledge bases.} {To analyze, locate and remove redundancies in KB is a way to refine and improve the efficiency of KBS. These type of analysis are important in topics such as probabilistic reasoning, information filtering, etc. (see also \cite{lang} for a general overview).}

The basic mechanism to obtain conservative retractions consists of eliminating, step by step, the variables that we wish to eliminate from language. Such a mechanism is called variable forgetting. Since the first analysis in cognitive robotics \cite{forget}, the problem of  variable forgetting is a widely studied technique in IA. In particular
the forgetting variable technique has been used to update or refine (logical, rule-based, CSP) programs. For example for resolution-based reasoning in specialized contexts (see e.g. \cite{Bledsoe}),  in CSP and optimization \cite{practical}, for simplification of rules \cite{Moinard} (included Answer Set Programming \cite{Eiter}). As it can be seen from these references, the interest in techniques for forgetting variables is not limited to classical (monotonous) logics. It has also received attention in the field of non-monotonous reasoning (including the computational complexity of the problems). In the (epistemic) modal logics for (multi)agency this technique would be very useful to represent knowledge-based games \cite{S5,vK}. In addition, in the reasoning under inconsistency the use of variable forgetting allows to weaken the KB to obtain consistent subKBs (eliminating the variables involved in the inconsistency). This topic, discussed below, is studied as a tool for solving SAT. In this context, providing methods for variable forgetting is a step towards the availability of retraction algorithms for programming paradigms based on logics of different nature.

{Taking into account the aforementioned motivations,} {\ttp the aim of this paper is twofold. {First}, we intend to present a complete and detailed exposition of the foundations of the {\em  independence rule} (That can be considered a tool for variable forgetting)},  whose basic ideas were published in \cite{calculemus}, since no detailed exposition has been published until now. In fact, we generalize the cited paper by stating the results for any operator that induces a conservative retraction, leading as consequence that our case,  the independence rule, is useful to compute the retractions. {\ttp It is also illustrated how the rule is useful to design methods for solving some questions on KB. In par\-ti\-cu\-lar the redundancy problem can be handled by the independence rule and Boolean derivatives (an essential tool to design this rule). }

{\ttp The structure of the paper is as follows. Section 2 is devoted to summarize the basics on the algebraic interpretation of propositional logics. In Sect. 3 the notion of forgetting operator is introduced, and we show how conservative retractions can be computed by means of these kind of operators, as well as logical calculus induced by them are sound and (refutationally) complete. Section 4 presents a forgetting operator inspired on the projection of algebraic varieties. The logical translation of the operator is presented as a inference rule in Sect. 5.  Boolean derivatives are used  for characterizing logical relevance (in the case of sensitive variables)  in algebraic terms (in section 6, Prop. \ref{sensitive}). Some illustrative applications of the tools presented are described in Sect. 7. The paper finishes with a discussion on the results presented in the paper as well as some ideas about future work.}

\section{Background}

\begin{figure}[t]\centering
  \includegraphics[width=11cm]{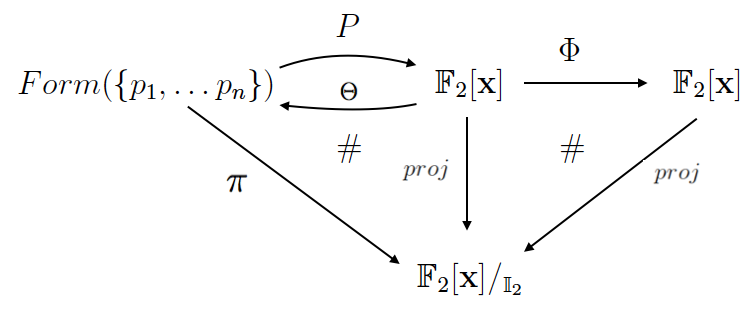}
  \caption{The framework\label{conmu}}
  \end{figure}

{In this section fundamental relations between propositional logic and polynomials with coefficients in finite fields (in our case, the finite field with two elements, $\cuerpo$) are summarized.}  {The main idea guiding the algebraic interpretation of logic is to identify a logical formula as a polynomial in such a way that the truth-value function induced by the formula could be understood as a polynomial function on $\cuerpo$.}

{The diagram showed above (Fig. \ref{conmu}) depicts the relationship between both structures, whose elements will be detailed in the following subsections. The ideal $\idos := \langle x_1+x_1^2,\dots , x_n+x^2_n \rangle \subseteq \anillo$ -on which we will talk about later- is used, and the map $proj$ is the natural projection on the quotient ring.} The remain elements of the diagram will be detailed bellow.

We assume throughout the paper that the reader is familiar with propositional logic as well as with basic principles on polynomial algebra on positive characteristics. 
\subsection{Propositional logic and conservative retraction}


A  propositional language is a finite set $\mathcal L=\{p_1,\dots,p_n\}$ of propositional symbols (also called propositional variables).  The set of formulas
$Form(\mathcal L)$ is built up from  in
the usual way, using the standard connectives $\lnot,\land,\lor,\lif$ and $\top$  ( $\top$ denotes the  constant {\em true}, and $\bot$ is $\lnot \top$). Given two formulas $F,G$ and $p\in \mathcal L$, we denote $F\{p/G\}$ the formula obtained replacing every occurrence of $p$ in $F$ by the formula $G$.

An interpretation (or valuation) $v$ is a  function $v: \mathcal L \to \{0, 1\}$. 
An interpretation $v$ is a {\bf model} of  $F\in Form (\mathcal L )$  if it makes $F$ true in the usual classical truth functional
way. {\ttp Analogously, it is said that a $v$ is a model of  $K$ ($v\models K$) if $v$ is model of  every formula in $K$. }
We denote by $Mod(F)$  the set of models
of $F$ (resp. $Mod(K)$ for the set of models of $K$).

A formula $F$  (or $K$ a KB) is {\bf consistent} if  it exhibits at least one model. It is said that  $K$ entails $F$ ($K\models F$) if  every model of $K$ is a model of $F$, that is, $Mod(K) \subseteq Mod(F)$.  Both notions can be naturally generalized to a KB, preserving the same notation. 
It is said that  $K$ and $K'$ are equivalent, $K' \equiv K$, if  $K\models K'$ and $K'\models K$. {The same notation will also be used for the equivalence with (and between) formulas.}

  It is said that $  K$ is an {\bf extension} of  $  K '$ if  ${\mathcal L}(  K')\subseteq {\mathcal L}(  K)$ and
 $$\forall F \in Form({\mathcal L}(  K')) [K' \models F \Lif K \models F ]$$
 
 $  K$ is a {\bf conservative extension} of  $  K '$ (or $  K'$ is a
 {\bf conservative retraction} of $  K$) if it is an extension such that every logic
 consequence of $  K$ expressed in the language ${\mathcal L}(  K ')$ is also consequence of $  K '$, 
 $$\forall F \in Form({\mathcal L}(  K')) [K \models F \Lif K' \models F ]$$
{\ttp  that is, $K$ extends $K'$ but  no new knowledge expressed by means of $\mathcal L (K')$ is added by $K$.}
 
Given  ${\mathcal L}' \subseteq {\mathcal L}(  K)$, a conservative retraction on the language ${\mathcal L}'$  always exists. The {\bf canonical conservative retraction}   of $K$ to $\mathcal L'$ is defined as:
$$[  K,\mathcal L']= \{F \in \mbox{Form}({\mathcal L}') \ : \    K \models F \}$$
That is, $[  K,\mathcal L']$ is the set of $\mathcal L'$-formulas which are entailed by $K$. In fact any conservative retraction on $\mathcal L'$ is equivalent to $[  K,\mathcal L']$. The actual issue is to present a finite axiomatization of such formula set.

\subsection{Propositional logic and the ring  ${\mathbb F}_2[{\bf x}]$}

{The ring $\anillo$ is naturally chosen for working with algebraic interpretations of logic.}
To clarify the notation, an identification
$p_i\mapsto x_i$ (or $p\mapsto x_p$) between $\mathcal L$ and
the set of indeterminates is fixed.

Notation on polynomials is standard. Given $\alpha=(\alpha_1,\dots , \alpha_n) \in  \N^n$, 
let us define 
$|\alpha |:=\mbox{max}\{\alpha_1,\dots , \alpha_n\}$. By ${\bf x}^{\alpha}$ we denote the monomial $x_1^{\alpha_1}\cdots x_n^{\alpha_n}$. 
 The {\em degree of}  $a({\bf x})\in
\anillo$, is
 deg$_{\infty}(a({\bf x})):=$max$\{|\alpha| \ : \ {\bf x}^{\alpha}
\mbox{ is a monomial of } a\}$. If deg$_{\infty}(a({\bf x}))\leq 1$, the polynomial $a({\bf x})$
we shall denote a {\bf polynomial formula}. 
{\ttp It is defined $\degree_i{(a({\bf x}))}$ as the  
degree  w.r.t. $x_i$.}
  
  \subsection{Translation from formulas and vice versa}
  
%

The  translation of Propositional Logics into Polynomial Algebra  is based 
on the following translation (see Fig. \ref{conmu}, left diagram):

The map $P: Form({\mathcal L}) \to \anillo$ is defined by: 
  \begin{itemize}
\item $ P(\bot)=0,  P(p_i)=x_i,   P(\lnot F)=1+P(F)$
\item $P(F_1\land
F_2)=P(F_1)\cdot P(F_2)$ 
\item $P(F_1 \lor F_2)=P(F_1)+P(F_2)+P(F_1)\cdot P(F_2)$
\item     $P(F_1 \to F_2)=1 + P(F_1)+P(F_1)\cdot P(F_2)$, and 
\item $P(F_1\liff
  F_2)=1+P(F_1)+P(F_2)$ 
  \end{itemize}

{For the reciprocal translation (from poynomials to formulas) we use the map $\Theta: \anillo \to Form({\mathcal L})$ defined by:}
   \begin{itemize}
\item $\Theta(0)=\bot ,\ \Theta(1)=\top$, $\Theta(x_i)=p_i,$ 
\item $\Theta(a\cdot b)=\Theta(a)\land \Theta(b)$, and $\Theta(a+b)=\lnot(\Theta(a)\liff \Theta(b))$.
   \end{itemize}
   It can be proved that $\Theta (P(F))\equiv F$ and  $P(\Theta (a))=a$. 
Sometimes, for the sake of readability, 
we will use the following property, that is a straightforward consequence of the previous assertions:
   $$\Theta (1+a+ab)\equiv \Theta(a) \to \Theta (b)$$

  \subsection{Correspondence between valuations and points in $\fdos^n$}

{The similar functional behavior of the formula $F$ and its polynomial translation $P(F)$  is the basis of the relationship between logical semantics and polynomial functions. Let's clarify what similar behavior means:}

\begin{itemize}
\item {\em From valuations to points}:
Given a valuation $v:\mathcal L \to \{0,1\}$,
the truth value of $F$ with respect to $v$ agrees with the value of $P(F)$ on
the  point of $o_v\in {\mathbb F}^n$ defined by the values provided by $v$: if 
 $(o_v)_i=v(p_i)$ then

$$v(F)=P(F)((o_v)_1,\dots (o_v)_n)$$ 

\item {\em From points to valuations}: Each $o=(o_1,\dots , o_n)\in \fdos^n$ induces a valuation $v_o$ defined by:
$$v_o(p_i)=1 \Liff o_i=1$$
\end{itemize}

This way
$$v_o \models F \Liff P(F)(o_v)+1=0 \Liff o_v \in V(1+P(F))$$
{where $V(.)$ is the well-known algebraic vanishing operator (see e.g. \cite{Cox}:} given $a({\bf x}) \in \anillo$, 
$$V(a({\bf x}))=\{o \in \fdos^n \ : \ a(o)=0\}$$

{Summarizing we provide two {maps} among the set of valuations and  points of $\fdos^n$, which are bijections between models of the formula $F$  and points from the algebraic variety determined by $1+P(F)$;}
$$\begin{array}{ccc}
\begin{array}{rcl} Mod(F) &\to & V(1+P(F))\\
v & \mapsto & o_v \\
\end{array}
   & \qquad &
   \begin{array}{rcl}
      V(1+P(F))&  \to&  Mod(F)\\
 o & \mapsto & v_o \\
 \end{array}
\end{array}
$$

For example, 
consider the formula $F= p_1\to p_2 \land p_3$. The associated polynomial is $P(F)=1+x_1+x_1x_2x_3$. The valuation $v=\{(p_1,0), (p_2,1), (p_3,0)\}$  is model of $F$ and induces the point $o_v=(0,1,0)\in \fdos^3$, which belongs to $V(1+P(F))=V(x_1+x_1x_2x_3)$.

\subsection{Polynomial projection}

 {\ttp Consider now the right-hand side diagram of  Fig. \ref{conmu}.} {To simplify the re\-la\-tion between the semantics of propositional logic and geometry over finite fields we use the map }
 $$\Phi : \anillo \to \anillo$$ 
$$\displaystyle \Phi(\sum_{\alpha \in I} {\bf x}^{\alpha}):=\sum_{\alpha
  \in I} {\bf x}^{sg(\alpha )}$$
 being $sg(\alpha ):=(\delta_1,\dots , \delta_n)$, where $\delta_i$ is $0$ if
$\alpha_i= 0$ and $1$ otherwise.

 The map $\Phi$ selects the representative element of the equivalence class of the polynomial in $\anillo/_{\idos}$ that is  a polynomial formula. So to associate a polynomial formula to a propositional formula $F$ it suffices to apply the composition  $\pi :=\Phi\circ P$, that we will call  {\bf polynomial projection}.  
For example, 
$$P(p_1 \to p_1 \land p_2)=1+x_1 + x_1^2x_2 \mbox{ whereas } \pi (p_1 \to p_1 \land p_2)=1+x_1 + x_1x_2$$

\subsection{Propositional Logic and polynomial ideals}

We recall here the well-known correspondence between algebraic sets and polynomial ideals on the coefficient field $\fdos$, and propositional logic KBs.

{Given a subset $X \subseteq (\mathbb{F}_2)^n$, we denote by 
$I(X)$  the set (actually an algebraic ideal) of polynomials of $\anillo$ vanishing on $X$:}

$$I(X)=\{a({\bf x})\in
\anillo \ : \ a(u)=0 \mbox{ for any }  u\in X \} $$

{Symmetrically}, given $J\subseteq \anillo$ it is possible to consider the previously mentioned algebraic set $V(J)$, the ``vanishing set":

$$V(J)=\{u \in  (\mathbb{F}_2)^n \ : \  a(u)=0 \mbox{ for any } a({\bf x})\in
J \}$$

Nullstellensatz theorem for $\fdos$ is  stated as follows (see e.g. \cite{janl}):
\begin{theorem} (Nullstellensatz theorem with the coefficient field $\fdos$)
\begin{itemize}
\item If $A\co {\mathbb F}_2^n$, then $V(I(A))=A$, and  
\item for every  $J\in Ideals({\mathbb F}_2[{\bf x}])$,
  $I(V({J}))=J+\idos$. 
\end{itemize}
\end{theorem}

From the Nullstellensatz theorem it follows that:
\smallskip
 \begin{center}
   $F\equiv F'$ if and only if $P(F)=P(F') \ (\mod \idos )$ 
 \end{center}
\smallskip
Therefore 
$F\equiv F'$ if and only if  $\pi (F) = \pi (F')$.

The following theorem summarizes the main relationship between propositional logic and $\anillo$:

\begin{theorem}  \label{gb} (see e.g.
\cite{jalonso}) Let $K=\{F_1,\dots , F_m\}$ and $G$ be a propositional formula.
The following conditions are equivalent:
\begin{enumerate}
\item $\{F_1,\dots , F_m\} \m G$.
\item $1+P(G) \in \langle 1+P(F_1),\dots , 1+P(F_m)\rangle + \idos$. 
\item $V\langle 1+P(F_1),\dots , 1+P(F_m)\rangle \subseteq V\langle 1+P(G)\rangle$
\end{enumerate}
\end{theorem}

\begin{remark}{\ttp 
If  the use of Gröbner basis is considered, above conditions are equivalent to:

 $ 4. \ \ \displaystyle {\tt NF}(1+P(G), 
\langle 1+P(F_1),\dots ,1+P(F_m)\rangle+\idos)=0$

\noindent where $\tt {GB(I)}$ denotes the {\em Gr\"{o}bner basis} of ideal
$I$ and $\tt {NF(p,B)}$
denotes a {\em normal form} of polinomial $p$ respect to the Gröbner basis $B$. 
The complete description on Gröbner basis is {not within the} scope of this paper. A general reference for Gröbner Basis could be seen in \cite{Winkler}. 
Readers can find in \cite{jalonso2}  a quick tour on the use of Gröbner Basis in Propositional Logic.}
\end{remark}

 

 
Given $K$ be a KB, let us define the ideal 
$$J_K = ( \{1+P(F) \ : \ F\in K\})$$
Note that by Thm. \ref{gb} it is easy to see that
$$ v\models K \Liff o_v \in V(J_K)$$

\section{Conservative retractions by forgetting variables}

{In this section  we present how to calculate a conservative retraction by using {\em forgetting} operators. These operators are maps of type:

$$ \delta : Form(\mathcal L)\times Form(\mathcal L) \to Form(\mathcal L)$$
 


\begin{definition} Let be $\delta$ an operator: $ \delta : Form(\mathcal L)\times Form(\mathcal L) \to {Form(\mathcal L\setminus \{p\})}$. It is said that $\delta$ is
\begin{enumerate}
\item {\bf sound} if $\{F,G\}\models \delta (F,G)$, and 
\item a {\bf forgetting operator} for the variable $p\in \mathcal L$ if
$$ \delta (F,G) \equiv [\{F,G\}, \mathcal L\setminus \{p\}]$$
\end{enumerate}
\end{definition}

\begin{figure}[t]
\includegraphics[width=9cm]{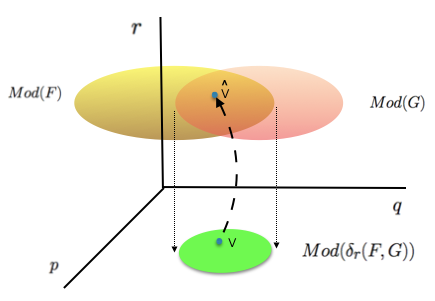}
\centering\caption{Semantic interpretation of a forgetting operator for the variable $r$ (Lifting Lemma) \label{indemod}}
\end{figure}

{An useful characterization of the operators can be deduced from the following semantic property:} {\ttp If $\delta$ is a forgetting operator, the models of $\delta (F,G)$ are precisely the {\em projections} of models of  $\{F,G\}$ (see Fig. \ref{indemod}).}

\begin{lemma} (Lifting Lemma) Let $v: \mathcal L \setminus \{p\} \to \{0,1\}$ be a valuation, $F,G\in Form (\mathcal L)$ and $\delta$ a forgetting operator for $p$. The following conditions are equivalent:
\begin{enumerate}
\item $v\models \delta (F,G)$
\item There exists a valuation $\hat{v}:\mathcal L \to \{0,1\}$ such that  $\hat{v}\models F\land G$ and 
$$\displaystyle \hat{v}\restriction_{\mathcal L \setminus \{p\}} = v$$
(that is, $\hat{v}$ extends $v$).
\end{enumerate}
\end{lemma}
\begin{proof}{\ttp
$(1)\Lifn (2)$: Given a valuation $v$,  let us consider the formula 
$$\displaystyle H_v = \bigwedge_{q\in \mathcal L\setminus \{p\}} q^v$$
where $q^v$ is $q$ if $v(q)=1$ and $\lnot q$ in other case. 
It is clear that $v$ is the only valuation on $\mathcal L \setminus \{p\}$ which is model of $H_v$.

\smallskip

Suppose that there exists a model of $\delta (F,G)$, $v: \mathcal L \setminus \{p\} \to \{0,1\}$,   with no  extension to a model of $F\land G$.  In this case the formula
$$H_v\to \lnot (F\land G)$$
 is a tautology, in particular
$$\{F,G\} \models H_v \to \lnot (F\land G)$$
Since $\{F,G\} \models  F\land G$, by modus tollens $ \{F,G\} \models \lnot H_v$. So $\delta(F,G)\models \lnot H_v$ because $\delta$ is a conservative retraction. This fact   is a contradiction because $v\models \delta(F,G) \land H_v$.

\medskip

$(2)\Lifn (1)$: {Such an extension $\hat{v}$ verifies}
$$\hat{v}\models F\land G \models [\{F,G\}, \mathcal L\setminus \{p\}]\models \delta (F,G)$$

Since $\delta (F,G)\in Form(\mathcal L\setminus \{p\})$, the valuation $v=\hat{v}_{\mathcal L \setminus \{p \}}$ is  also a model of $\delta (F,G)$.}

  \end{proof}
  
 In particular the result is true for the canonical conservative retraction $[K,\mathcal L \setminus \{p\}]$, because  
  $$
  [ K , \mathcal L \setminus \{p \} ] \equiv \delta_p (\Land K, \Land K)
  $$
  (being $\Land K := \Land_{F\in K} F$)

  An interesting case appears when $\delta_p (F_1, F_2)\equiv \top$. In this case every partial valuation on $\mathcal L \setminus \{p\}$ is
  extendable to a model of $\{F_1,F_2 \}$.

The following characterization will be used later:

\begin{corollary}\label{cli} Let $ \delta : Form(\mathcal L)\times Form(\mathcal L) \to {Form(\mathcal L\setminus \{p\})}$ be a sound operator. The following conditions are equivalent:
\begin{enumerate}
\item $\delta$ is a forgetting operator for the variable $p$.
\item For any $F,G\in Form (\mathcal L)$ and  $v \models \delta (F,G)$ valuation on $\mathcal L \setminus \{p\}$,  there exists an extension of $v$ model of $\{F,G\}$.
\end{enumerate}
\end{corollary}
\begin{proof}

$(1) \Lifn (2)$: Is true by Lifting Lemma

$(2) \Lifn (1)$. Let $F,G$ be two  formulas.   {\ttp Since $\delta$ is sound, it suffices to see that  
$$\delta(F,G) \models  [\{F,G\}, \mathcal L\setminus \{p\}]$$}

{\ttp Suppose that is not true. In this case there exists $H\in Form(\mathcal L\setminus \{p \})$ such that $[\{F,G\}, \mathcal L\setminus \{p\}]\models H$ (so $\{F,G\}$ also entails $H$), but there exists a valuation $v$ satisfying  $v\models \delta(F,G)\land \lnot H$. 

By $(2)$ there exists $\hat v$ extension of $v$ which is model of $\{F,G\}$, so
$\{F,G\} \not\models H$, that is, a contradiction.}


 \end{proof}

 \begin{corollary}\label{basicos} If  $p\notin var(F)$, and $\delta_p$ is a forgetting operator for $p$, then 
 $$\delta_p(F,F) \equiv F  \qquad \mbox{and } \qquad  \delta_p(F,G) \equiv \{F , \delta_p(G,G)\}$$
 \end{corollary}
 \begin{proof} 
 If $p\notin var(F)$, then $\{F\} \equiv [ \{F\}, \mathcal L \setminus  \{p\}]\equiv \delta_p(F,F)$
 
 On the other hand, $\delta_p(F,G)\equiv [\{F,G\}, \mathcal L \setminus \{p\}] \models \{F , \delta_p(G,G)\}$. {To prove that actually it  is an equivalence, it will be shown that they have the same models.
  
  Let $v$ a valuation on $\mathcal L \setminus \{p\}$  such that $v \models \{F,  \delta_p(G,G)\}$.} {\ttp Then there exists $\hat{v}$, extension of $v$, such that $\hat{v}\models G$. Since $\hat{v}\models F$, then by Lifting lemma $v\models \delta (F,G)$.}

  \end{proof}

For forgetting operators as defined herein, the Lifting Lemma is a reformulation  of the observation made by J. Lang et al. \cite{lang} about variable forgetting. The authors  present a characterization of forgetting by means of Quantified Boolean Formulas (QBF), 
$\exists x \hat{F} (x)$, where $\hat{F}$ is the interpretation of $F$ as a Boolean formula whose free variables are the propositional variables of $F$. In our case,  $\delta_p(F,G)$ could correspond to the QBF formula $\exists  p (F\land G)$.

{Authors of the aforementioned article present a method of forgetting $X$ (a variable set of a formula $F$), denoted by $forget(F,X)$ by constructing disjunctions in the following way:}
\begin{itemize}
\item[ ] $forget(F,\emptyset)= F$
\item[ ] $forget(F,\{x\}) = F\{x/\top\}\lor F\{x/\bot\}$
\item[ ] $forget(F, \{x\}\cup Y )=forget(forget(F,Y),\{x\})$
\end{itemize}

Note that with this approach $forget(F, Y )$ can have high size. In our case we aim to simplify the representation by using algebraic operations on polynomial projections. 



\subsection{Conservative retractions induced by forgetting operators}
\vspace{0.1cm}
We denote by $2^X$ the power set of $X$. By analogy with the classical resolution-based saturation process (on CNF formulas), we will call {\em saturation} the process of applying the rule exhaustively until no new consequences are obtained and observing the result  (checking whether an inconsistency has been obtained) .
\begin{definition}
\mbox{ }
\begin{enumerate} 
\item Let $\delta_p$ be a forgetting operator for $p$. It is defined $\delta_p[\cdot] $ as 
$$\displaystyle\delta_p[\cdot] : 2^{Form(\mathcal L)} \to 2^{Form(\mathcal L)}$$
$$\delta_p[K]:=\{\delta_p(F,G) \ : \ F,G \in K\}$$
\item Suppose we have a forgetting operator $\delta_p$ for each $p\in \mathcal L$. We will call {\bf saturation} of $K$ to the process of applying the operators $\delta_p[\cdot]$ (in some order) by using all the propositional variables of $\mathcal L (K)$, denoting the result by $sat_{\delta}(K)$ (which will be a subset of $\{\bot, \top \}$).
\end{enumerate}
\end{definition}

We will bellow see that the set $sat_{\delta}(K)$ does not essentially depend of the order of applications of operators. Moreover, keep in mind that since the forgetting operators are sound, if $K$ is consistent then necessarily $sat_{\delta}(K)=\{\top\}$.

From forgetting operators a logical calculus can be defined in the usual way:
\begin{definition}\label{calculus} Let $K$ be a KB and $F\in Form({\mathcal L})$ and let $\{\delta_p \ : \ p \in \mathcal L (K) \}$ a family of forgetting operators.
\begin{itemize}
\item A $\te_{\delta}$-proof in  $K$ is a formula sequence $F_1,\dots F_n$ such that for every $i\leq n$  $F_i\in K$ or exist $F_j,F_k$ ($j,k<i$) such that $F_i=\delta_p(F_j,F_k)$ for some $p\in \mathcal L$.
\item $K \te_{\delta} F$ if there exists $\te_{\delta}$-proof in $K$, $F_1,\dots F_n$, with $F_n=F$
\item A $\te_{\delta}$-refutation is a $\te_{\delta}$-proof of $\bot$.  
\end{itemize}
\end{definition}

The (refutational) completeness of the calculus associated to forgetting operators is stated as follows.

\begin{figure}[t]
\centering\includegraphics[width=10cm]{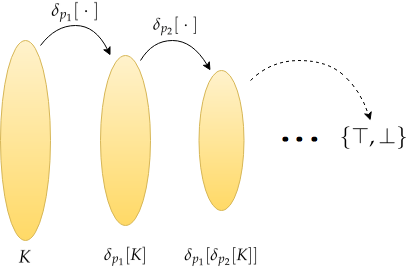}
\caption{Deciding consistency by using a set of forgetting operators $\partial$ \label{comple}}
\end{figure}

\begin{theorem}\label{completa}  Let $\{\delta_p \ : \ p \in \mathcal L \}$ a family of forgetting operators. Then $\te_{\delta}$ is refutationally complete:
$K$ is inconsistent if and only if $K \te_{\delta} \bot$. 
\end{theorem}
\begin{proof} The idea is to saturate the KB (Fig. \ref{comple}).  If $sat_{\delta}(K)=\{\top\}$, then, by repeating the application of Lifting Lemma, we can extend the empty valuation (which is model of $\{\top\}$) to a model of $K$

If $\bot \in sat_{\delta}(K)$ then $K$ is inconsistent, because $K\models sat_{\delta}(K)$ by soundness of forgetting operators.
The selection of a particular $\te_{\delta}$-refutation is straightforward, as in the proof of the refutational completeness of resolution calculus, for example.
\end{proof}

\begin{corollary}\label{key}
$\delta_p[K]\equiv [K,\mathcal L \setminus \{p\}]$
\end{corollary}
\begin{proof}
By soundness of  the forgetting operator $\delta_p$,
$$[K,\mathcal L \setminus \{p\}] \models \delta_p[K]$$
\noindent holds.
To prove  the other direction, let $F\in [K,\mathcal L \setminus \{p\}]$, and let us suppose that   $\delta_p[K] \not\models F$. Then $\delta_p[K] + \{\lnot F\}$ is consistent. In particular, if we saturate, $sat_{\delta}(\delta_p[K] \cup \{\lnot F\})=\{\top \}$. 

Since $p\notin var(\lnot F)$, by Lemma \ref{basicos} it holds  that for any $G\in K$:
$$\delta_p(\lnot F,G) \equiv \{\lnot F,  \delta_p(G,G) \} \qquad \mbox{and } \qquad \delta_p(\lnot F, \lnot F) \equiv \lnot F$$
Therefore
$$\delta_p[K\cup \{\lnot F\}] \equiv \delta_p[K] \cup \{\lnot F\}$$
so, by applying saturation, starting with  $p$
$$sat_{\delta}(K\cup \{\lnot F\})\equiv sat_{\delta}(\delta_p[K] \cup \{\lnot F\})=\{\top\}$$
 what indicates that $K\cup \{\lnot F\}$ is consistent, thereupon $K\not\models F$, a contradiction.

\end{proof}

 Given
 $Q\co \mathcal L$ and a  linear order $q_1< \cdots < q_k$ on $Q$, we define
the operator
$$\delta_{Q,<}:= \delta_{q_1}\circ \cdots \circ \delta_{q_k}$$
In fact, if we dispense with making an order explicit, the operator is well-defined module logical equivalence, that is, any two orders on $Q$ produce equivalent KBs. This is true because  for each $p,q\in \mathcal L$, using the previous corollary it follows that 
 $$\delta_p \circ \delta_q [K ]\equiv \delta_q \circ
\delta_p [K ]$$ 
\noindent due to the fact that both KBs are equivalent to $[K, \mathcal L \setminus \{p,q\}]$. Therefore, for the sake of simplicity, we will write $\delta_Q[K]$ when syntactic presentation of  this KB does not matter.

A consequence of corollary \ref{key} and theorem \ref{completa} is that entailment problem can be reduced to a similar problem but that it only uses
variables of the goal formula. This property is called the {\em location
property}:  the entailment problem can be simplified by eliminating propositional variables that do not appear in
the target formula.

\begin{corollary}\label{entail} (Location Property, \cite{calculemus}) The following conditions are e\-qui\-va\-lent:
  \begin{enumerate}
\item $K \m F$
\item  $\delta_{\mathcal L \setminus var(F)}[K ]\m F$
\end{enumerate}
\end{corollary}
\begin{proof}
It is trivial, because $\delta_{\mathcal L \setminus var(F)} [K]  \equiv [ K, var(F)]$.
 \end{proof}

\section{Boolean derivatives and independence rule on polynomials}


 In order to define our forgetting operator we will make use of derivations on polynomials, by translating the  usual derivation  on $\anillo$ to
an operator  on propositional formulas. We review here some basic
  properties.
Recall that a derivation on a
  ring $R$ is  a
  map
$d: R \to R$ verifying 
$$d(a+b)=d(a)+d(b) \mbox{ and } d(a\cdot b)= d(a)\cdot b + a\cdot d(b) \qquad \mbox{for any }a,b\in R$$

The logical translation of derivations is builded as follows:
\begin{definition}\cite{calculemus}
A map   $\partial: Form({\mathcal L}) \to Form({\mathcal L})$ 
is a {\em Boolean derivative}  if there 
 exists a derivation $d$ on ${\mathbb F}_2[{\bf x}]$
such that
 the following diagram is commutative:
 \begin{center}
   $
\begin{array}{ccccc}
      Form({\mathcal L}  )          &       & \stackrel{\partial}{\to}        &  & Form({\mathcal L})       \\
   \\
  \pi \downarrow  &             &     \#         & &  \uparrow  \Theta \\
  \\
     {\mathbb F}_2[{\bf x}] &   & \stackrel{d}{\to}       &  & 
{\mathbb F}_2[{\bf x}]  \\ 
                             \end{array}
$
 \end{center}
That is, 
$\partial = \Theta \circ d \circ \pi$
\end{definition}

In this paper we are particularly interested in the Boolean derivative, denoted by 
$\frac{\partial}{\partial p}$, induced by the derivation  $d=\frac{\partial
}{\partial x_p}$. The following result shows a semantic equivalent expression of this derivative.

\begin{proposition}\label{derivada}
$\frac{\partial}{\partial p}F\equiv \lnot (F\{p/\lnot p\}\liff F)$
\end{proposition}
 \begin{proof} It is straightforward to see that
$$\pi (F\{p/\lnot p\})({\bf x})= \pi (F)(x_1,\dots ,
x_p+1 , \dots, x_n)$$
On the other hand it is easy to see that
$$\frac{\partial}{\partial x}a(x)=a(x+1)+a(x)$$ holds for polynomial formulas, hence
$$\frac{\partial}{\partial x_p}\pi(F) = \pi (F)(x_1,\dots ,
x_p+1 , \dots, x_n) + \pi(F) (x_1,\dots ,
x_p, \dots, x_n)$$
Therefore, by applying $\Theta$ we conclude that
$$\frac{\partial}{\partial p} F = \Theta (\frac{\partial}{\partial x_p}\pi(F)) \equiv \lnot (F\{p/\lnot p\}\liff F)
$$

 \end{proof}

 {Notice that truth value of $\frac{\partial}{\partial p} F$ with
respect to a valuation does not depend of the truth value on the own  $p$; hence, we
can apply 
 valuations on $\mathcal L \setminus \{p \}$  to this formula. In fact, we can describe the structure of $F$ by isolating the role of $p$ as follows: }
 
 \begin{lemma}\cite{calculemus}\label{desc} ($p$-normal form).
  Let $F\in Form({\mathcal L})$ and $p$ be a propositional variable. There exists
 $F_0\in Form({\mathcal L\setminus\{p\}})$ such that
 
$$F\equiv \lnot (F_0 \liff p\land \frac{\partial}{\partial p}F)$$
\end{lemma}
\begin{proof} Since $\pi(F)$ is a polynomial formula, we can suppose that
 $$\pi (F) = a + x_p b  \mbox{ with } deg_{x_p}(a)=deg_{x_p}(b)=0$$
 Therefore
 $$ F \equiv \Theta (\pi (F)) \equiv \lnot (\theta (a) \liff p \land \Theta (b)) $$
Then let $F_0 = \Theta (a)$, and, since $b= \frac{\partial}{\partial x_p}\pi(F)$,  we have that $\Theta (b) = \frac{\partial}{\partial p} F$.
 \end{proof}

For example, let $F= p\land q \to r$. Then $$\pi (F)= 1+x_px_q+x_px_qx_r=1+x_p(x_q+x_qx_r)$$
 Following the above proof, $a=1$ and $b=x_q+x_qx_r$ (note that $\frac{\partial}{\partial x_p} (\pi(F))=x_q+x_qx_r$). Therefore $\Theta(a)=\top$ and $\Theta(b)=q\land \lnot r$, so
$$ F\equiv \lnot (\top \liff p \land (q\land \lnot r))$$

\begin{figure}[t]
\includegraphics[width=10cm]{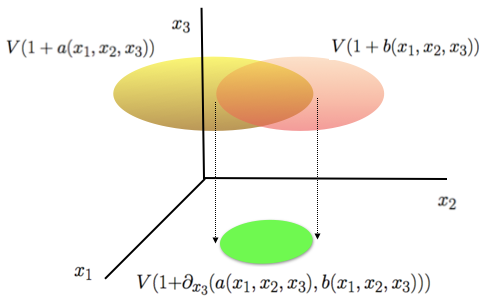}
\centering\caption{Geometric  interpretation of independence rule \label{indegeo}}
\end{figure}

 {The forgetting operator that we are going to define next, called independence rule, aims to represent the models of the conservative retraction as those that can be extended to models of $ F \land G $ (that is, the idea behind Lifting Lemma). Geometrically, if $a$ and $b$ are the polynomials $\pi(F)$ and $\pi (G)$ respectively, then the vanishing set $V(1+a, 1+b)$ (which could correspond to the set of  models both of $F$ and $G$) is projected by $\partial_{p}$ (see Fig. \ref{indegeo}). The algebraic expression of the projection is described as a rule.
}

\begin{definition} 
    The {\em independence rule} (or 
 $\partial$-rule) on polynomial formulas  is defined as follows: given $a_1,a_2\in \anillo$ and $x$ an indeterminate
$$\displaystyle\hspace*{.5cm}
 \quad \frac{a_1, \ a_2  }{\partial_x(a_1,a_2)}
 $$
 where $\partial_x(a_1,a_2)={1+\Phi\left[(1+a_1\cdot a_2)(1+a_1\cdot\frac{\partial}{\partial x}a_2
    +a_2\cdot \frac{\partial }{\partial x}a_1
+ \frac{\partial }{\partial x}a_1 \cdot \frac{\partial }{\partial
  x}a_2)\right]} $
\end{definition}
If $a_i=b_i+x_p \cdot c_i, \ \mbox{with }
\degree_{x_p}(b_i)=\degree_{x_p}(c_i)=0 \ (i=1,2)$, 
the rule we can rewritten as: 
$$ \partial_{x_p}(a_1,a_2)={\Phi\left[1+(1+b_1\cdot b_2)[1+(b_1+c_1)(b_2+c_2)]\right]}$$
\hspace{0cm} 

 For example, to compute 
 $$a=\partial_{x_2}(1+x_2x_3x_5+x_3x_5, 1+x_1x_2x_3x_4x_5 + x_1x_2x_3x_5)$$
we take
$$b_1= 1+ x_3x_5, \ c_1 = x_3x_5 \qquad \mbox{and } \qquad b_2 =1, \ c_2 = (1+x_4)x_1x_3x_5$$
so  the result is $a=1+x_1x_3x_4x_5 + x_1x_3x_5.$

 Note that independence rule is symmetric.

\section{Independence rule and non-clausal theorem proving}

The {\bf independence rule} for formulas is defined as
\begin{center}
  $\partial_{p}(F,G):= \Theta(\partial_{x_p}(\pi (F),\pi(G)))$
\end{center}
Following with above example,

\begin{table}[h]
\centering
\begin{tabular}{l}
$\partial_{p_2}(p_3\land p_5 \lif p_2 , p_1 \land p_2 \land p_3 \land p_5 \lif p_4) =$ \\\\
\qquad $= \Theta(\partial_{x_2}(1+x_2x_3x_5+x_3x_5, 1+x_1x_2x_3x_4x_5 + x_1x_2x_3x_5)) = $ \\\\
\qquad $= \Theta(1+x_1x_3x_4x_5 + x_1x_3x_5) = \lnot (p_1\land p_3\land p_4\land p_5 \liff p_1\land p_3 \land p_5) \equiv$ \\\\
\qquad $\equiv p_1\land p_3 \land p_5 \to p_4$
\end{tabular}
\end{table}

%
%

 It is worthy to {point out} some interesting features of the rule $\partial_p$: if $\partial_p(F,G)$ is a tautology, then $\partial_p(F,G)=\top$, and if $\partial_p(F,G)$ is inconsistent then $\partial_p(F,G)= \bot$. Both features are consequence of the translation to polynomials: polynomial formulas corresponds of tautologies and inconsistencies  are algebraically simplified to $1$ and $0$ in $\anillo/\idos$, respectively. In fact, we will usually work with the polynomial projections to exploit these features. 



\begin{proposition}
$\partial_{p}$ is sound
\end{proposition}
\begin{proof}
{We have to prove $F_1\land F_2\models \partial_p (F_1,F_2)$. Suppose that 
$$\pi(F_1) =b_1+x_p \cdot c_1, \quad \pi(F_2) =b_2+x_p \cdot c_2$$  

According to Thm. \ref{gb}.(3),  it is enough to prove that
$$V(1+ \pi(F_1)\cdot \pi(F_2)) \subseteq V(1+\partial_{x_p}(\pi(F_1),\pi(F_2)))$$

Let ${\bf u} \in V(1+ \pi(F_1)\cdot \pi(F_2))\subseteq \fdos^n$, that is,
$$(b_1+x_pc_1)(b_2+x_pc_2)\large|_{\bf x = u}=1 \qquad \qquad (\dag)$$}
(the notation used here is as usual: $F(x)\large|_{x = u}$ is $F(u)$).
{\ttp Let us distinguish two cases:
\begin{itemize}
\item If the $p$-coordinate of  $\bf u$ is $0$, then by  $(\dag)$ it follows that 
$$b_1\large|_{\bf x = u}=b_2\large|_{\bf x = u}=1$$ 
Therefore $(1+b_1b_2)\large|_{\bf x = u}=0$. 
\item The $p$-coordinate of  $\bf u$ is $1$. In this case $(b_1+c_1)(b_2+c_2)\large|_{\bf x = u}=1$
\end{itemize}
By examining the definition of $\partial_p$ we conclude  in both cases  that}
$$ \partial_{x_p} (\pi(F_1),\pi(F_2))\large|_{\bf x = u}=1$$
so we have that  ${\bf u}\in V(1+ \partial_{x_p} (\pi(F_1),\pi(F_2)))$.
\end{proof}

\begin{theorem}\label{induce}
$\partial_{p}$ is a forgetting operator
\end{theorem}
\begin{proof}
{The goal is to prove that
$$[\{F_1,F_2\}, \mathcal L \setminus \{p \}] \equiv \partial_p(F_1,F_2)$$
  Let us suppose that $F_1,F_2 \in Form(\mathcal L)$ such that $\pi (F_i)=b_i + x_p c_i$ $i=1,2$ with $b_i,c_i$ polynomial formulas without variable $x_p$. Recall that in this case the expression of the rule is
$$\partial_{x_p}(\pi(F_1),\pi(F_2))= \Phi \left( (1+b_1 \cdot b_2) \left[1+ (b_1+c_1)(b_2+c_2)\right] \right)$$

Since the soundness of $\partial_p$ has been proved by the previous proposition, by Corollary \ref{cli} it is sufficient to show that any valuation $v$ on $\mathcal L \setminus \{p\}$ model of $\partial_p(F_1,F_2)$ can be extended to $\hat v\models \{F_1,F_2\}$. 

Let $v \models \partial_p(F_1,F_2)$. Let us consider the point from $\fdos^n$ asociated to $v$, $o_v$. It follows that
  \begin{eqnarray}
    o_v\in  V(\pi(\partial_p(F_1,F_2))+1) =  V(\partial_{x_p} (\pi(F_1),\pi(F_2))+1) =  \nonumber\\
 =   V((1+b_1 \cdot b_2) [1+ (b_1+c_1)(b_2+c_2)]) \nonumber
    \end{eqnarray}
    
    so  $$\displaystyle \left((1+b_1 \cdot b_2) [1+ (b_1+c_1)(b_2+c_2)]\right)\large|_{\bf x = o_v}=0$$
    
    In order to build the required extension $\hat v$, let us distinguish two cases:
\begin{itemize}
\item If $(1+b_1 \cdot b_2)\large|_{\bf x = o_v}=0$ then  $\hat v= v\cup \{(x_p, 0)\}\models F_1 \land F_2$.

\item If $[1+ (b_1+c_1)(b_2+c_2)]\large|_{\bf x = o_v}=0$ then $$\hat v= v\cup \{(x_p, 1)\}\models F_1 \land F_2$$
\end{itemize}
}

 \end{proof}

{With some abuse of notation, we use the same symbol, $\te_{\partial}$, to denote similar notions that defined in  Def. \ref{calculus} but on polynomial formulas and rules  $\partial_{x_p}$. In that way we can describe $\te_{\partial}$-proofs  on polynomials.  For example, a
$\partial$-refutation  for the set
$\pi[\{p\lif q , q\lor r \lif s , \lnot (p\lif s)\}]$ is}
\begin{enumerate}
\item $1+x_1+x_1x_2 \hfill [\![\pi(p\lif q)]\!]$
\item  $1+(x_2+x_3+x_2x_3)(1+x_4) \hfill [\![\pi(q\lor r \lif s)]\!]$
\item  $x_1(1+x_4) \hfill [\![ \pi(\lnot(p\lif s) ]\!]$
\item  $1+x_1+x_3+x_1x_4+x_3x_4+x_1x_3+x_1x_3x_4 \hfill [\![\partial_{x_2}\ to 
    (1),(2)]\!]$
\item  $0 \hfill [\![\partial_{x_1}\ to\ (3),(4)]\!]$
\end{enumerate}



\begin{corollary}\cite{calculemus}

 $K$ is inconsistent if and only if $K \te_{\partial} \bot$. 
\end{corollary}
\begin{proof} It is consequence of theorems \ref{completa} and \ref{induce}

 \end{proof}

The result, in algebraic terms, is as follows:

\begin{corollary}\label{prop2}
 Let $F\in Form({\mathcal L})$ and let $K$ be  a knowledge basis. The following conditions are
  equivalent:
\begin{enumerate}
\item  $K \m F$
\item $J_K\te_{\partial} 0$
\end{enumerate}
\end{corollary} 
\begin{proof}

$(1)\Lifn (2)$:
Let us suppose $K\models F$. Then $K+ \{\lnot F\}$ is inconsistent. Since $\partial_p$ is refutationally complete, $K+ \{\lnot F\}\te_{\partial}\bot$. Thus
 $$ \{1 + \pi (G) \ : \ G\in K\} \cup \{\pi(F) \} \te_{\partial} 0$$
 
 $(2)\Lifn (1)$: If a $\partial$-refutation is founded on polynomials, then by above theorem $K \cup \{\lnot F\}$ is inconsistent.
  \end{proof}

 \begin{remark}
 To compute conservative retractions we use an implementation (in Haskell language) of $\partial_p$ and $\partial_{x_p}$. In order to simplify the presentation, we only show the computation on polynomials (that is, the application of $\partial_{x_p}$), and we use the own propositional variables as polynomial variables (that is, we identify $p$ and $x_p$) to facilitate the readibility.
 
The software used in the examples and experiments can be downloaded from \url{https://github.com/DanielRodCha/SAT-Pol}

 \end{remark}
 
\begin{example}  Let  $G=s\to r$ and $K$ be the KB
$$
K = \left\{
\begin{array}{l}
t \land p \liff s \\
t \land r \to s \\
t \land q \to s \\
p \land q \land s \land t \to r \\
\end{array}
\right.
$$

 To decide whether $K\models G$ -by applying location lemma- we have to compute  
 $$\partial_{\mathcal L \setminus \{r,s\}}[K]\equiv \partial_p[\partial_q[\partial_t[K]]] $$
{\tt 
\begin{itemize}
\item[*] [pqrst+pqst+1,pt+s+1,qst+qt+1,rst+rt+1]     \hfill    (projection)
\item[*] [pqrs+pqs+ps+s+1, ps+s+1, 1]                         \hfill   (forgetting t)
\item[*] [ps+s+1, 1]                                                         \hfill  (forgetting q)
\item[*] [1]                                                                        \hfill (forgetting p)
\end{itemize}
}
\noindent Therefore:

$$[K, \mathcal L \setminus \{r,s\}] 
\equiv \{\top \}\not\models G
$$ 

Consider now the formula $F = p\land q \land t \to s$. In order to decide whether $K\models F$,  by location lemma we have to compute $ [K, \mathcal L (F)]\equiv\partial_r[K]$
{\tt 
\begin{itemize}
\item[*]     [pqrst+pqst+1,pt+s+1,qst+qt+1,rst+rt+1] \hfill     (projection)
\item[*]     [pqst+pqt+pt+qst+qt+s+1,pt+s+1,qst+qt+1,1]  \hfill     (forgetting r)
\end{itemize}
}
To see that $\partial_r[K]\models F$ it is sufficient to show that $\partial_r[K] \cup \{\lnot F\}$ is inconsistent. The computation is made in the projection set, by saturating the polynomial set:
{\tt 
\begin{itemize}
\item[*]               [pt+s+1,qst+qt+1, pqst+pqt]          \hfill   (the retraction and $\lnot$ F)
\item[*] [0] \hfill (applying $sat_{\partial}$)
\end{itemize}
}
 \end{example}

\smallskip

{An approach to specify contexts in AI for reasoning is to determine which set of variables $Q\subseteq K$ provides information and which variables are irrelevant for represent the specific context.  In fact, in some approaches for formalizing context-based reasoning contexts are determined by this variable set. When $K$ does not provide any specific information about the context in which it is to be used, it is natural to conclude $[K, Q]$ should only contain tautologies, that is, $\partial_{\mathcal L \setminus Q}[K]=\{\top\}$.  
In the previous example $K$ does not provide relevant information about the context determined by  $\{r, s\}$, because $[K, \mathcal L \setminus \{r,s\}] =  \{\top \}$}.

%
%
%
%
%
%
%


\section{Characterization 
  sensitive implications}

{In addition to its use in the design and study of the independence rule, other use of Boolean derivatives is the detection of variables that are irrelevant in a formula (or, in terms of \cite{for}, to study when a formula is independent of a variable), and more generally, when a variable is irrelevant in a formula relativized to a KB (that is, in the models of the own KB).}

{We will say that a variable $p$ is {\bf irrelevant} in a formula $F$ (or $F$ is independent of $p$) if $F$ is equivalent to a formula in which $p$ does not occur. This concept can be generalized to a set of variables in the natural way, and it can be proved that a $F$ formula is independent of a set of variables $X$ if and only if $F$ is independent from each variable of $X$ (see \cite{for}). In this paper also remarks the following result:}
\begin{proposition} The following conditions are equivalent:
\begin{enumerate}
\item $F$ is independent from $x$
\item $F\{x/\top\} \equiv F\{x/\bot\}$
\item $F\{x/\top\} \equiv F$
\item $F\{x/\bot\} \equiv F$
\end{enumerate}
\end{proposition}
To which we could add:
\begin{enumerate}
\item[5.] $\models \lnot \frac{\partial}{\partial_p}(F)$ 
\end{enumerate}

{In this section we are interested in studying the notion of independence relativized to a KB. Note that it may happen that a variable may be relevant in a formula but not in the KB models we are working with. To distinguish the relativized notion from the original we will use the word {\em sensitive}.}

\begin{definition}
A formula $F$ is called {\bf sensitive in $p$ with respect to} a knowledge basis $K$
if $K\not\m F\{p/\lnot p\}\liff F$. We say that $F$ is sensitive w.r.t. $K$
(or simply sensitive, if $K$ is 
fixed)  if $F$ is sensitive in all its variables. 
\end{definition}

The following result habilitates the use of Gröbner basis for determining sensitiveness (by means of ideal membership test in condition (4)) or that of our interest, by means of $\partial_p$-rules (condition (3)).

{\ttp It is straightforward to check that:}
\begin{proposition} \label{sensitive}
Let $p\in$var$(F)$. The following conditions are equivalent:
\begin{enumerate}
\item  $F$ is {\em sensitive} in $p$ with respect to a knowledge basis $K$
\item $K\cup \{\frac{\partial}{\partial_p}(F)\}$ is consistent
\item  $sat_{\partial}[ K\cup \{\frac{\partial}{\partial_p}(F)\}] =\{\top \}$
\item $ \partial_{x_p} \pi(F) \notin J_K + \idos$ 
\end{enumerate}
\end{proposition}
\begin{proof} 

$(1) \Lif (2)$: Since $K\not\models F\{p/\lnot p\} \liff F$, there exists $v\models K$ where $v\models \lnot(F\{p/\lnot p\} \liff F)$, that is, $v\models K\cup \{\frac{\partial}{\partial_p}(F)\}$

$(2) \Lif (3)$ by completeness of $\te_{\partial}$

$(3) \Lif (4)$: Let $v\models K\cup \{\frac{\partial}{\partial_p}(F)\}$ (it is consistent by $(3)$. Then $\pi(\frac{\partial}{\partial_p}(F))(o_v)=1$. 

 Moreover $\pi(G)(o_v)=1$ for any $G\in K$ hence $o_v\in J_K$. Therefore the polynomial $\frac{\partial}{\partial_p}\pi(F)$ does not belong to $J_K$


$(4) \Lif (1)$: Suppose  $\frac{\partial}{\partial_{x_p}}\pi(F) \notin J_k + \idos$ so there exists $o\in V(J_k)$ such that   $\frac{\partial}{\partial_{x_p}}\pi(F)(o)\neq 0$. Then $v_o \models \frac{\partial}{\partial_p}(F)$ and $v_o\models K$. Therefore $K\cup \{\frac{\partial}{\partial_p}(F)\}$ is consistent, so $F$ is sensitive in $p$ w.r.t. $K$.

%
%
%

\end{proof}

{From the definition itself it follows that Boolean derivatives can be used to tackle the problem of sensitive arguments in implications}: $F$ is not sensitive in $p$ w.r.t. $K$  iff $K\m \lnot
\frac{\partial}{\partial p} F$.
 In this case,
there exists  $G$ with
var$(G)=$var$(F)\setminus \{p \}$ such that $K\m F\liff G$
(e.g. $F\{p/\bot\}$).

\begin{example}\label{KK} {Lets us consider the following consistent KB as a rule-based system}
{\ttp
$$ K = \left\{
\begin{array}{ll}
R1 : p_1 \to p_9 \\
R2 : p_1 \to p_{10} \\
R3 : \lnot p_2 \to p_9  \\
R4 : \lnot p_2 \to p_{10}  \\
R5 : (p_1 \land p_7) \to p_{11} \\
R6 : p_3 \to p_7\\
R7 : p_3 \to p_{10}\\
R8 : p_4 \to p_{11}\\
R9 : p_5 \to p_8\\
R10 : p_6 \to p_9\\
\end{array}\right.
$$
{Let us consider as a set of potential facts (potential inputs of the system) }
$$\mathcal F = \{p_1, \dots , p_6, \lnot p_1, \dots , \lnot p_6\}$$
{We will say that a rule} $R\in K$ is sensitive in $p$ w.r.t. to  $K$ and a subset of potential facts $\mathcal C$ if  $R$ is sensitive in $p$ w.r.t. $K\cup \mathcal C$. Let us compute some examples:
\begin{itemize}
\item $R1$ is sensitive in $p_1$ w.r.t. $K$ and the potential fact set $\{\lnot p_2\}$.
$$\frac{\partial}{\partial_{p_1}}(R1) = \theta (\frac{\partial}{\partial_{x_1}}(1+x_1(1+x_9))=\lnot p_9$$
and
$$K\cup\{\lnot p_2\} \not\models \lnot  p_9$$ 
This  condition can be checked (by using condition $(3)$ of above proposition) showing that  
$$[K\cup\{\lnot p_2\}, \{p_9\}]\equiv \partial_{\mathcal L \setminus \{ p_9\}}[K\cup\{\lnot p_2\}]=\{p_9\}\models \lnot p_9$$
The computation is:
{\tt \small
\begin{itemize}
\item[*] [p1p10+p1+1,
           p1p11p7+p1p7+1,
           p1p9+p1+1, p10p2+p10+p2,
           
           p10p3+p3+1,
           p11p4+p4+1,
           p2p9+p2+p9,
           p2+1,
           p3p7+p3+1,
           
           p5p8+p5+1,
           p6p9+p6+1] \hfill (projection of $K\cup\{\lnot p_2\}$)
\item[*] [p10p2+p10+p2,p10p3+p3+1,p11p4+p4+1,p2p9+p2+p9,
            p2+1,
            
            p3p7+p3+1,
            p5p8+p5+1,p6p9+p6+1,1] \hfill (forgetting p1)
\item[*] [p10p3+p3+1,p10,p11p4+p4+1,p3p7+p3+1,p5p8+p5+1,
           
           p6p9+p6+1,p9,1] \hfill (forgetting p2)
\item[*] [p10,p11p4+p4+1,p5p8+p5+1,

p6p9+p6+1,p9,1] \hfill (forgetting p3)
\item[*] [p10,p5p8+p5+1,

p6p9+p6+1,p9,1]  \hfill (forgetting p4)
\item[*] [p10,p6p9+p6+1,p9,1] \hfill (forgetting p5)
\item[*] [p10,p9,1] \hfill (forgetting p6)
\item[*] [p10,p9,1] \hfill (forgetting p7)
\item[*] [p10,p9,1] \hfill (forgetting p8)
\item[*] [p9,1] \hfill (forgetting p10)
\item[*]  [p9,1] \hfill (forgetting p11)
\end{itemize}

}
\item $R5$ is not sensitive in $p_1$ w.r.t. the set $\{p_4\}$:
$$\frac{\partial}{\partial_{p_1}}(R_5)=\Theta(\frac{\partial}{\partial_{x_1}}(1+(x_1x_7)(1+x_{11})))=\Theta(x_7(1+x_{11}))=p_7\land \lnot p_{11}$$
and $K\cup\{p_4\}\not\models \frac{\partial}{\partial_{p_1}}(R_5)$. In fact

$$[K\cup\{p_4\} , \{p_7,p_{11}\}]\equiv  \partial_{\mathcal L \setminus \{ p_7 , p_{11}\}}[K\cup\{p_4\}] =  \{p_{11}\}\models \lnot  \frac{\partial}{\partial_{p_1}}(R_5)$$
\end{itemize}
 
}
\end{example}

\section{Some applications}

{We will illustrate the usefulness of the tools presented to work on different types of KRR-related tasks. For reasons of paper length we will only describe two of them.}

\subsection{Detecting potentially dangerous states}
{In \cite{amai} authors show an algebraic method
for detecting potentially dangerous states in a Rule Based Expert System (RBES) whose
knowledge is represented by Propositional  Logic. Given $K$ made up of rules, the idea is to consider the potential facts that make $K$ to infer an unwanted value (which leads a danger or undesirable  state). Formally, they specify both the set $\mathcal F$ of potential facts  (potential input literals of the RBES) and dangerous states. 

Let us consider the following example, taken from the same \cite{amai}, to show that it avoids dangerous situations and  $K$ from example \ref{KK}.}
{Let 

$$\mathcal F = \{p_1, \dots , p_6, \lnot p_1, \dots , \lnot p_6\}$$ 

\noindent and let  $p_{11}$ be a warning variable of a dangerous state. We know the initial state the information $\{p_1,\lnot p_2\}$, which is a secure information because $K\cup \{p_1,\lnot p_2\}\not \models p_{11} $. 

In this case the question  is to detect which potential facts lead us to that dangerous state, i. e. which literals $r\in \mathcal F$ verifying that 
$$K\cup \{p_1,\lnot p_2\} \cup \{ r\} \models p_{11}$$
 By deduction theorem it is equivalent to decide whether 
$$K\models p_1 \land \lnot p_2  \land r \to p_{11}$$
To apply Location Lemma it is sufficient to bear in mind that this will be true if and only if

\smallskip

\hspace*{2cm} $[K, \{p_1,\dots p_6 \} \cup \{p_{11}\}]\models p_1 \land \lnot p_2  \land r \to p_{11}$ \hfill $(\dag \dag)$

\smallskip

In this case, $$K'=[K, \{p_1,\dots p_6 \} \cup \{p_{11}\}] \equiv  \partial_{\mathcal L \setminus (\{p_1, \dots , p_6, p_{11}\})}[K]$$ In poynomial terms, this $KB$ is represented as $$J=\{\tt p1p11p3+p1p3+1,p11p4+p4+1,1\}$$


In order to use the results of this paper it suffices to use Deduction Theorem for reducing condition $(\dag \dag)$ to

$$[K, \{p_1,\dots p_6 \} \cup \{p_{11}\}]  \cup \{p_1,\lnot p_2\} \models r \to p_{11}$$

\smallskip

\noindent and, by Reductio ab absurdum we have to find variables $r$ from $\mathcal F$ such that

$$[K, \{p_1,\dots p_6 \} \cup \{p_{11}\}]  \cup \{p_1,\lnot p_2,\lnot p_{11}\} \models \lnot r$$

Obviously $p_1, p_2$ do not satisfy this.
\begin{itemize}
\item For $i=3,4$ we have $[K',\{p_i\}]\equiv\{\lnot p_i\}$ (because the polynomial projection, computed using $J$, is $\tt 1+pi$), so they are  dangerous variables
\item For $i=5,6$ we have $[K',\{p_i\}]\equiv\{\top\}$, (because the polynomial projection is $\tt 1$), so they are not dangerous variables
\end{itemize}


 


\subsection{Decomposing KB for context-based reasoning}

This example is taken from \cite{Amir}. Let us suppose we analyze the behavior of an espresso machine whose functioning aspects are captured by the axioms enumerated in Fig. \ref{espresso}. The first four axioms denote that if the machine pump is OK and the pump is on, then the machine has a water supply. Alternately, the machine can be filled manually, but this never happens when the pump is on. The next four axioms denote that there is some steam if and only if the boiler is OK and it is on and there is enough water supply. The last three axioms denote that there is always either coffee or tea, and that steam and coffee (or tea) result in a hot drink.  Let 
the knowledge base for an espresso machine described in Fig. \ref{espresso}. {Although the example is very simple, it is very useful to show the applicability of the paper tools.  We have relied on the cases of \cite{Amir}. }

\begin{figure}[t]
\begin{center}
\includegraphics[width=12cm]{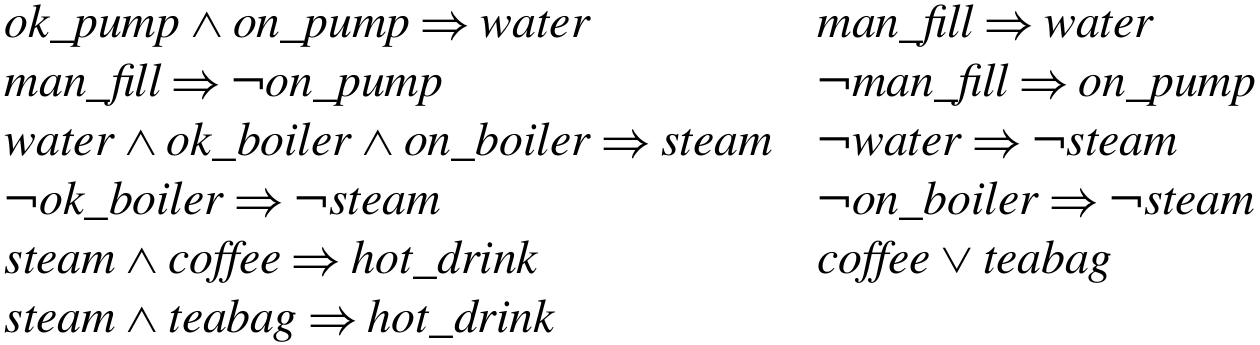}
\end{center}
\caption{Knowledge Base $\mathcal A$ for an espresso machine from \cite{Amir}\label{espresso}}
\end{figure}


{Authors in \cite{Amir} obtain the partition described in  Fig. \ref{part}. The languages from these partitions can be used to compute
 conservative retractions, in order to check the completeness of the refined KB obtained in the above-mentioned paper. We compute the conservative retractions by using operators $\partial_p$ defined in this paper.

\begin{figure}[t]
\begin{center}
\includegraphics[width=14cm]{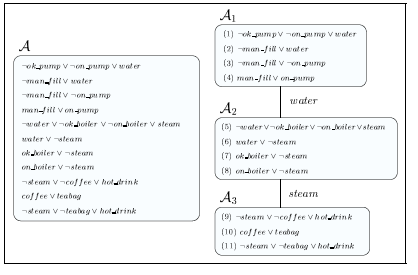}
\end{center}
\caption{Partition of Knowledge Base of Fig. \ref{espresso}, extracted from \cite{Amir}\label{part}}
\end{figure}

\begin{itemize}
\item $\mathcal L (\mathcal A_1)=\{on$\_$pump, ok$\_$pump, water, man$\_$fill\}$. The analogous one in our approach, $[\mathcal A , \mathcal L (\mathcal A_1)]$, is equivalent to $\mathcal A_1$,

$$[\mathcal A , \mathcal L (\mathcal A_1)]=\left\{
\begin{array}{l}
ok$\_$pump \land on$\_$pump \to water, \ man$\_$fill \to water, \\
 man$\_$fill \to \lnot on$\_$pump, \  \lnot man$\_$fill \to on$\_$pump
 \end{array}
 \right\} $$

\item $\mathcal L (\mathcal A_2)=\{water, on$\_$boiler, ok$\_$boiler, steam\}$.
The analogous one in our approach, $[\mathcal A , \mathcal L (\mathcal A_2)]$, is also equivalent $\mathcal A_2$,

$$[\mathcal A , \mathcal L (\mathcal A_2)]=\left\{
\begin{array}{l}
 steam \to ok$\_$boiler, \
steam \to on$\_$boiler,\
steam \to water,\\
ok$\_$boiler \land on$\_$boiler\land water \to steam\end{array}
 \right\}$$

\item Finally $\mathcal L (\mathcal A_3)=\{steam, coffee, hot$\_$drink, teabag\}$. The analogous one in our approach, $[\mathcal A , \mathcal L (\mathcal A_3)]$, is also equivalent to $\mathcal A_3$,

$$[\mathcal A , \mathcal L (\mathcal A_3)]=\left\{
\begin{array}{l}
coffee\land steam \to hot$\_$drink, \ coffee \lor  teabag, \\
steam\land teabag \to hot$\_$drink, 
\end{array}
 \right\}$$
\end{itemize}

\section{Experiments}

The proposal of this paper is of a general nature and is not specialized
in specific fragments of propositional logic. Nevertheless, in this
section we are going to experimentally compare saturation based on the
independence rule with saturation based on the basic rule of forgetting
variables (see \cite{lang}), which can be applied to
propositional logic without syntactic restrictions, to which we add a
simplification operator (we call {\em canonical} to this rule). We will formally see
this idea in the next subsection. In adition, before describing the experimental results, we will show in the
second subsection a refinement of the retraction to decrease the number
of rule applications (for any of them).

\subsection{Canonical forgetting operator}

A specific syntactic feature of the forgetting operator $\partial_p$  is that,  if $\partial_p(F,G)$ is a tautology, then $\partial_p(F,G)=\top$, and if $\partial_p(F,G)$ is inconsistent then $\delta_p(F,G)= \bot$. This characteristic is a consequence of the pre- and post-processing of formulas by means of polynomial translations and vice versa. Under this translation, tautologies and inconsistencies  are algebraically simplified to $\top$ and $\bot$ respectively. 
It would not be true for any forgetting operator in general, but at least you can achieve some simplification by eliminating occurrences of $\top, \bot$,  producing thus only reduced formulas (with no occurrences of $\top$ and $\bot$). For this purpose we use a {\em simplification operator}:
$$\sigma : Form(\mathcal L) \to Form^r (\mathcal L)$$
(where $Form^r (\mathcal L)$ is the set of  reduced formulas) defined as:
\begin{enumerate}
\item $\sigma(s)=s $ if $s\in  \{\top, \bot\}$, $\sigma(\lnot \top)=\bot$ and $\sigma(\lnot \bot)=\top$, 
\item $\sigma (F)= F$ if $\bot$, $\top$ do not occur in $F$, and in other case:
\begin{enumerate}
 \item $\sigma (\top \land F)= \sigma(F)$, $\sigma (\top \lor F)= \top$
 \item $\sigma (\bot \land F)= \bot$   and $\sigma (\bot \lor F)= \sigma(F)$
\item $\sigma (\top \lif F)= \sigma(F)$, $\sigma (F \lif \top)= \top$, $\sigma (\bot \lif F)= \top$ and $\sigma (F \lif \bot)= \sigma(\lnot \sigma(F))$
\item If $F,G\neq \top,\bot$, $\sigma(F*G)=\sigma(\sigma(F)*\sigma(G))$  for $*\in \{\land,\lor, \lif\}$ and $\sigma(\lnot F)=\sigma(\lnot \sigma(F))$
\end{enumerate}
\end{enumerate}
For example, 
$$\sigma((\bot \to p)\land (\top \land q))= \sigma (\sigma(\bot \to p)\land \sigma(\top \land q))= \sigma(\sigma(p)\land\sigma(q))=\sigma(p\land q  )=p\land q$$

It is straight to see that $\sigma\circ \delta \equiv \delta$ for any operator $\delta$. 

\begin{definition} The {\bf canonical forgetting operator} for a variable $p$ is defined as $$\delta_p^0=\sigma \circ \delta^*_p$$
where $$\delta_p^* (F,G) := (F\land G)\{p/\top\}\lor (F\land G)\{p/\bot\}$$
\end{definition}

\begin{proposition} $\delta_p^0$ is a forgetting operator for $p$
\end{proposition}
\begin{proof}  Easy  by the  Lifting Lemma
\end{proof}

\begin{example} $F=p\to q$ and $G=p\land r \to \lnot q$:
$$\begin{array}{ll}
\delta_p^0(p\to q, p\land r \to \lnot q) & = \sigma(\delta^*_p(p\to q, p\land r \to \lnot q) \\
							& =  \sigma \left( [(p\to q) \land (p\land r \to \lnot q)]\{p/\top\} \lor \right. \\ 
							& \left. \qquad [(p\to q )\land (p\land r \to \lnot q)]\{p/\bot\}\right) = \\
							& =  \sigma \left( [(\top \to q) \land (\top \land r \to \lnot q)] \lor \right. \\ 
							& \left. \qquad [(\bot\to q )\land (\bot\land r \to \lnot q)]\right) = \\
							& =  \sigma \left( \sigma [(\top \to q) \land (\top \land r \to \lnot q)] \lor \right. \\ 
							& \left. \qquad  \sigma[(\bot\to q )\land (\bot\land r \to \lnot q)]\right) = \\
							& =  \sigma \left( [q \land (r \to \lnot q)] \lor \top \right) = \top\\
\end{array}
$$
\end{example}

\begin{corollary} Let $ \delta : Form(\mathcal L)\times Form(\mathcal L) \to Form(\mathcal L\setminus \{ p\})$.  The following conditions are equivalent:
\begin{enumerate}
\item $\delta$ is a forgetting operator for $p$.
\item $\delta\equiv \delta^0_p$
\end{enumerate}
\end{corollary}

\subsection{Refining the process}

In order to make the implementation of the realistic algorithms, we will use the following result that significantly reduces the number of applications of the variable forgetting rules in practice. Besides, using the fact that it's symmetrical, we can even further reduce the number of applications of the operators:

\begin{proposition} In above conditions
$$
\delta_p[K] \equiv \{ F \ : \ p \mbox{ does not occur in } F \} \cup
 \delta_p[\{F \in   K \ : \ p \mbox{ occurs in } F \}]$$
\end{proposition}
\begin{proof}
 Let us denote by $A$  the first set of formulas and by $B$  the second one. We just have to prove that  $A\cup B\models \delta_p[K]$.

Let $\delta_p(F,G)$ be a formula of $\delta_p[K]$. By symmetry, it is enough to consider only three cases:
\begin{itemize}
\item $p\notin var(F) \cup var(G)$. Then $A\models F\land G \equiv\delta_p(F,G) $
\item $p\in var(G) \setminus var(F)$. Then $A\cup B \models F\land \delta_p(\top, G) \equiv \delta_p(F,G) $
\item $p\in  var(F) \cap var(G)$. Then $\delta_p(F,G) \in B$
\end{itemize}

\end{proof}

\subsection{Experimental results}

In this section, as an illustration, we will execute variable forgetting on a set of knowledge bases to show the efficiency of the rule proposed in this paper with respect to the canonical rule in the processing of forgetting operations (fundamentally, we will focus on the cost in space, the number of symbols used in the representation). Each experiment has been performed by randomly choosing some variables present in the knowledge base and order on them (common for both operators) and we will progressively apply the corresponding operators\footnote{Although it is irrelevant to calculate the size of the knowledge bases, to estimate the time we have used a MacBook Air with a 1.6 GHz Intel Core i5 processor and 8 GB 1600 MHz DDR3 memory. The operating system is macOS High Sierra 10.13.5}. Below we describe the datasets and the results obtained.

The examples are taken from the SAT Competition 2018 website\footnote{\url{http://sat2018.forsyte.tuwien.ac.at}}. In Fig. \ref{inicio} their initial size (both in propositional logic and polynomial transformed) is shown. The total time used in each dataset experiment (i.e. in the application of all the corresponding operators chosen for that experiment) for both the proposed and the canonical operator is also shown. In the following tables the results of the progressive implementation of the operators are shown.

\begin{figure}[t]
{\footnotesize
\begin{tabular}{|l|r|r|r|r|r|r|}\hline 
Name                                               & Size      & Size                    & seconds & Space (bytes)   & seconds  & Space (bytes) \\ 
                                                          &  form.  & pol.  & Can. rule  & Can.             & Indep.  & Indep.  \\ \hline \hline
{\scriptsize mp1-bsat180-648}  &  14715 & 18176 & 6.86&  1,417,335,160  & 3.94  &  1,407,227,832   \\ \hline
{\scriptsize unsat250} & 24746 &  31800 & 74.60 & 26,453,571,424  & 0.92  &  864,549,664   \\ \hline
{\scriptsize mp1-squ\_any\_s09x07\_c27\_bail\_UNS} & 36619 & 60318 & 410.87& 81,535,687,624  & 5.24  &  1,706,513,448   \\ \hline
{\scriptsize g2-modgen-n200-m90860q08c40-13698} & 236112 & 319913 & 31.55& 16,500,421,928  & 6.58  &  5,153,583,936   \\ \hline
{\scriptsize mp1-klieber2017s-1000-023-eq} & 351346 & 416531 &  out of time & -- & 518.07& 426,862,970,928   \\ \hline
{\scriptsize mp1-tri\_ali\_s11\_c35\_bail\_UNS} & 37606 & 43493  & 4.34& 1,510,720,256  & 0.48  &  498,118,016   \\ \hline
{\scriptsize mp1-Nb5T06} & 211656 & 252250 & 75.89 & 38,378,129,872 & 1.03  &  1,858,302,896   \\ \hline
\end{tabular}
}
\caption{SAT instances used in the experiments: size (both formula and polynomial representation) and total cost using both rules \label{inicio}}
\end{figure}

\begin{center}
{\footnotesize mp1-bsat180-648} 
\end{center}
 \begin{multicols}{2}
      \begin{tikzpicture} 
       \begin{axis}[small,axis lines=middle,grid,
       legend pos= north west]
         \addplot+[no marks,red] coordinates{(1,15428)(2,16900)(3,25131)(4,26839)(5,34908)(6,36735)(7,54530)(8,55231)(9,57073)(10,169330)(11,384996)};
         \addlegendentry{\footnotesize{Canonical}}
         \addplot+[no marks,green] coordinates{(1,19757)(2,20876)(3,21866)(4,23516)(5,24612)(6,29269)(7,33490)(8,34985)(9,47857)(10,59215)(11,60313)};
          \addlegendentry{\footnotesize{Indep. rule}}
       \end{axis}
     \end{tikzpicture}      
    
\begin{scriptsize}
\begin{tabular}{|c|c|c|}
\hline
& Size of KB & Size of KB \\ 
Step & & (Indep.  \\ 
& (canonical) & rule)\\ \hline
1 & 15428 & 19757\\  
2 & 16900 & 20876 \\  
3 & 25131 & 21866 \\ 
4 & 26839 & 23516 \\ 
5 & 34908 &  24612\\ 
6 & 36735 & 29269 \\ 
7 & 54530 & 33490\\ 
8 & 55231 & 34985 \\ 
9 & 57073& 47857 \\ 
10 & 169330 & 59215 \\ 
11 & 384996 & 60313 \\ \hline
\end{tabular}     
\end{scriptsize}
\end{multicols}

\newpage

\begin{center}
{\footnotesize unsat250} 
\end{center}
\begin{multicols}{2}
      \begin{tikzpicture} 
       \begin{axis}[small,axis lines=middle,grid,
       legend pos= north west]
         \addplot+[no marks,red] coordinates{(1,26280)(2,28906)(3,34141)(4,68952)(5,191069)(6,193949)(7,200293)(8,214468)(9,231861)(10,407926)(11,939192)};
         \addlegendentry{\footnotesize{Canonical}}
         \addplot+[no marks,green] coordinates{(1,34091)(2,35443)(3,35891)(4,37731)(5,39242)(6,40393)(7,54677)(8,56846)(9,61119)(10,62330)(11,64986)};
          \addlegendentry{\footnotesize{Indep. rule}}
       \end{axis}
     \end{tikzpicture}      

\begin {scriptsize}
\begin{tabular}{|c|c|c|}
\hline
& Size of KB & Size of KB \\ 
Step & & (Indep.  \\ 
& (canonical) & rule)\\ \hline
1 & 26280  & 34091  \\  
2 & 28906  & 35443  \\  
3 & 34141& 35891  \\ 
4 &  68952 & 37731  \\ 
5 & 191069  & 39242  \\ 
6 & 193949  & 40393  \\ 
7 & 200293  & 54677  \\ 
8 & 214468  & 56846  \\ 
9 & 231861  &  61119 \\ 
10 & 407926  & 62330  \\ 
11 & 939192  & 64986  \\ \hline
\end{tabular}     
\end {scriptsize}
\end{multicols}

\begin{center}
{\footnotesize mp1-squ\_any\_s09x07\_c27\_bail\_UNS} 
\end{center}
\begin{multicols}{2}
      \begin{tikzpicture} 
       \begin{axis}[small,axis lines=middle,grid,
       legend pos= north west]
         \addplot+[no marks,red] coordinates{(1,37028)(2,37557)(3,38496)(4,53437)(5,411028)(6,589455)(7,731546)(8,825714)(9,890920)(10,995522)};
         \addlegendentry{\footnotesize{Canonical}}
         \addplot+[no marks,green] coordinates{(1,62498)(2,64058)(3,63921)(4,66035)(5,65907)(6,67631)(7,67512)(8,68884)(9,68774)(10,69832)};
          \addlegendentry{\footnotesize{Indep. rule}}
     \end{axis}
     \end{tikzpicture}      

\begin {scriptsize}
\begin{tabular}{|c|c|c|}
\hline
& Size of KB & Size of KB \\ 
Step & & (Indep.  \\ 
& (canonical) & rule)\\ \hline
1 &  37028 & 62498  \\  
2 & 37557  & 64058  \\  
3 & 38496 & 63921  \\ 
4 & 53437  & 66035  \\ 
5 & 411028  & 65907  \\ 
6 & 589455  & 67631  \\ 
7 & 731546 & 67512  \\ 
8 &  825714 & 68884  \\ 
9 & 890920  &  68774 \\ 
10 &  995522  &  69832  \\ \hline
\end{tabular}     
\end {scriptsize}
\end{multicols}

\begin{center}
{\footnotesize g2-modgen-n200-m90860q08c40-13698} 
\end{center}

\begin{multicols}{2}
      \begin{tikzpicture} 
       \begin{axis}[small,axis lines=middle,grid,
       legend pos= north west]
         \addplot+[no marks,red] coordinates{(1,237141)(2,242392)(3,290876)(4,319155)(5,3339545)(6,3340602)(7,3346440)(8,3362678)(9,3367621)(10,3390750)};
         \addlegendentry{\footnotesize{Canonical}}
         \addplot+[no marks,green] coordinates{(1,321926)(2,326062)(3,327171)(4,332485)(5,334528)(6,341128)(7,344213)(8,348381)(9,373298)(10,537880)};
          \addlegendentry{\footnotesize{Indep. rule}}
       \end{axis}
     \end{tikzpicture}      

\begin {scriptsize}
\begin{tabular}{|c|c|c|}
\hline
& Size of KB & Size of KB \\ 
Step & & (Indep.  \\ 
& (canonical) & rule)\\ \hline
1 & 237141  & 321926  \\  
2 & 242392  & 326062  \\  
3 & 290876 & 327171  \\ 
4 & 319155  & 332485  \\ 
5 & 3339545  & 334528  \\ 
6 & 3340602  & 341128  \\ 
7 & 3346440  & 344213  \\ 
8 & 3362678  & 348381  \\ 
9 & 3367621 & 373298  \\ 
10 & 3390750  & 537880  \\  \hline
\end{tabular}     
\end {scriptsize}
\end{multicols}

\begin{center}
{\footnotesize mp1-klieber2017s-1000-023-eq} 
\end{center}
\begin{multicols}{2}
      \begin{tikzpicture} 
       \begin{axis}[small,axis lines=middle,grid,
       legend pos= north west]
         \addplot+[no marks,red] coordinates{(1,351346)(2,1176815)(3,2233525)(4,7599271)(5,74792215)};
         \addlegendentry{\footnotesize{Canonical}}
         \addplot+[no marks,green] coordinates{(1,1615542)(2,28418531)(3,28418518)(4,28418513)(5,28419039)(6,28419304)};
          \addlegendentry{\footnotesize{Indep. rule}}
       \end{axis}
     \end{tikzpicture}      

\begin {scriptsize}
\begin{tabular}{|c|c|c|}
\hline 
& Size of KB & Size of KB \\ 
Step & & (Indep.  \\ 
& (canonical) & rule)\\ \hline
1 & 351346  & 1615542  \\  
2 & 1176815  &  28418531 \\  
3 & 2233525 &  28418518 \\ 
4 & 7599271  & 28418513  \\ 
5 & 74792215  & 28419039  \\ 
6 & out of time  & 28419304  \\ 
\hline
\end{tabular}     
\end {scriptsize}
\end{multicols}

\begin{center}
{\footnotesize mp1-tri\_ali\_s11\_c35\_bail\_UNS} 
\end{center}

\begin{multicols}{2}
      \begin{tikzpicture} 
       \begin{axis}[small,axis lines=middle,grid,
       legend pos= north west]
         \addplot+[no marks,red] coordinates{(1,37877)(2,38312)(3,38884)(4,50809)(5,54837)(6,213958)(7,237610)};
         \addlegendentry{\footnotesize{Canonical}}
         \addplot+[no marks,green] coordinates{(1,44127)(2,44943)(3,44771)(4,44634)(5,48200)(6,48072)(7,51134)};
          \addlegendentry{\footnotesize{Indep. rule}}
       \end{axis}
     \end{tikzpicture}

\begin {scriptsize}
\begin{tabular}{|c|c|c|}
\hline
& Size of KB & Size of KB \\ 
Step & & (Indep.  \\ 
& (canonical) & rule)\\ \hline
1 & 37877  &  44127 \\  
2 & 38312  & 44943  \\  
3 & 38884 & 44771  \\ 
4 & 50809  &  44634  \\ 
5 & 54837  & 48200  \\ 
6 & 213958  & 48072  \\ 
7 & 237610  & 51134  \\ 
\hline
\end{tabular}     
\end {scriptsize}
\end{multicols}

\begin{center}
{\footnotesize mp1-Nb5T06} 
\end{center}

\begin{multicols}{2}
   \begin{tikzpicture} 
       \begin{axis}[small,axis lines=middle,grid,
       legend pos= north west]
         \addplot+[no marks,red] coordinates{(1,212003)(2,213563)(3,222107)(4,222446)(5,224030)(6,232858)(7,233189)(8,234797)(9,243909)(10,294441)(11,3975580)(12,3975911)(13,3977425)(14,3985776)(15,4033048)(16,7353991)(17,22919499)};
         \addlegendentry{\footnotesize{Canonical}}
         \addplot+[no marks,green] coordinates{(1,252400)(2,252570)(3,252724)(4,252862)(5,253000)(6,253482)(7,253620)(8,254102)(9,254240)(10,254378)(11,254516)(12,254998)(13,255136)(14,255278)(15,255760)(16,255898)(17,256036)};
          \addlegendentry{\footnotesize{Indep. rule}}
       \end{axis}
   \end{tikzpicture}      

\begin {scriptsize}
\begin{tabular}{|c|c|c|}
\hline
& Size of KB & Size of KB \\ 
Step & & (Indep.  \\ 
& (canonical) & rule)\\ \hline
1 & 212003  & 252400  \\  
2 & 213563  &  252570 \\  
3 & 222107 & 252724  \\ 
4 & 222446  & 252862  \\ 
5 & 224030  & 253000  \\ 
6 & 232858  & 253482  \\ 
7 & 233189  & 253620  \\ 
8 & 234797  &  254102 \\ 
9 & 243909  & 254240  \\ 
10 & 294441  & 254378  \\ 
11 & 3975580  &  254516 \\
12 & 3975911  & 254998  \\
13 & 3977425  & 255136  \\
14 & 3985776  & 255278  \\
15 & 4033048  & 255760  \\
16 &  7353991 & 255898  \\
17 & 22919499  &  256036 \\ \hline
\end{tabular}     
\end {scriptsize}
\end{multicols}

As we have observed in the experiments, despite the fact that initially the transformation to polynomials uses slightly more space, the growth in the size of the knowledge base representation when progressively applied by the operators improves considerably the resources needed with respect to the canonical operator. Note, as we have already noted, that we have compared our operator (not restricted to a sub-language of the propositional logic) with the non-specialized operator.

\section{Conclusions and future work}
 

As we have already mentioned in the introduction, the algebraic interpretation of propositional logic represents a valuable bridge for applying algebraic techniques in KRR. In this paper we have proposed a new algebraic model to solve problems on KRR  whose  knowledge is represented by propositional  logic. We are not concerned here about the practical computational cost of the use of independence rule, the aim of a next paper. We focused on its theoretical foundations and potential applications in KRR instead. 

The main technique introduced in the paper is the use of a new rule inspired in the projection of algebraic varieties and polynomial derivatives. Also, {it is justified} that with Boolean derivatives {is possible to} determine  specific cases of conditional independence, the formula-variable independence {relativized to} a KB. The tools have been used  to solve questions related with distilling KB to obtain relatively simpler KBs to solve context-based questions. 

 Throughout the paper we have {remarked} some works related to  the tools used here. These works are driven to exploit the computation and use of Gröbner Basis, whilst in this paper a new method is proposed (specific for $\fdos$). 

With regard to the practical complexity of the proposed method, the application of the independence rule is reduced to the the algebraic simplification of polynomials. Its calculation appears at two levels: first in the multiplication of Boolean polynomials (or in finite fields in general), since the rule of independence is reduced to products, and second in the transformation of formulas into polynomials for a complete knowledge base.

With respect to the first question, the product of Boolean polynomials is a long-studied problem, for which refined algorithms exist (see e.g. \cite{faster}). The computational complexity of the application of the rule is irrelevant (it consists mainly of four polynomial products) compared to the second issue, the translation of formulas into polynomials. Such kind of transformation has been widely used to compile and run (using Gr\"{o}bner databases) rule-based programs,  for example. The problems where the polynomial interpretation of the formulas (rules, for example) of a KB is used are problems where the formulas have very limited complexity. 
In this type of program, the number of variables that appear in a rule is very small compared to the total number (see e.g. \cite{depre,fault,menu}). Therefore the computation of conservative retraction is feasible with our approach. Moreover, in the context of algebraic interpretation of logic reasoning,  the results shown in this article allow us to replace the use of "black box" implementations (those that use a computer algebra system to compile and reason) with another also algebraic but "white box" type, which can be verified/certified.
The complexity of algebraic simplifications can be high when both the number of variables is large and the number of variables that occur in each knowledge base formula is relatively high (with respect to the size of the total set of variables). However it is not common (nor advisable)  in programming paradigms such as logical programming (answer set or rule-based programming), DL ontologies, etc.

With respect to the use of independence rule as  SAT solver, although the rule is (refutationaly) complete, its intended use in this article is to calculate conservative retractions, and we have not presented a SAT algorithm other than the intuitive one (rule application saturation). In the GitHub repository \url{https://github.com/DanielRodCha/SAT-Pol}, variants are being developed to solve more efficiently SAT problems. With regard to the complexity of the problem of the variable forgetting in the case of propositional logic, it has been studied in considerable depth due to its relationship with the SAT problem (and, in general, with problems with Boolean functions) both for the foundations on the complete logic and for various fragments of this logic (see e.g. \cite{for,MSMP}).  As we have shown in the section devoted to experiments, computational cost of polynomials computations are irrelevant in practice, bearing in mind that our method is nos specialized for fragments of propositional logic, since it works on the full propositional logic. That section focuses on comparing the rule with the general rule of forgetting variables, since it does not seem appropriate to compare our proposal with other approaches specialized in fragments of propositional logic.

As future work we  will intend to work in two complementary research lines. On the one hand we intend to carry out the extension to many-valued logics
and their applications \cite{jalonso,jalonso2,polynomial} as well as to tackle the knowledge forgetting problem in modal logics \cite{S5}. 
If the underlying logic is many-valued, the algebraic varieties of the polynomial translations of the propositions do not behave intuitively (see \cite{racsam}). Therefore,  we need to use an appropiate version of the the Nullstellensatz Theorem. Also, for  this research line, a careful generalization of  the concept of Boolean derivatives,
with nice  logical meaning, has to be carried out  \cite{lebron}. On the other hand,   it is possible  to use  our model - in a similar way to the applications presented in the paper- for implementing expert systems based on the knowledge of different experts as in \cite{mcs15},  for  diagnosis (see \cite{fibro}) or -in the behalf of authors- more promising field of inconsistency management \cite{incon}.

\section{Acknowledgements}

This work was partially supported by  TIN2013-41086-P project (Spanish Ministry of Economy and Competitiveness), co-financed with FEDER funds. We also acknowledge the reviewers and editors for their valuable comments and suggestions, which have substantially improved the content of this article.

\section{References}

\bibliographystyle{elsarticle-num} 

\bibliography{biblio} 

\end{document}